\documentclass[11pt,a4paper]{article}

\usepackage[a4paper,margin=2.15cm]{geometry}
\usepackage{setspace}
\usepackage{amsmath,amssymb,amsthm,mathtools}
\numberwithin{equation}{section}
\usepackage{bm}
\usepackage{tensor}
\usepackage{slashed}
\usepackage{cancel}

\usepackage{booktabs}
\usepackage{array}
\usepackage[dvipsnames]{xcolor}
\usepackage{graphicx}
\usepackage{caption}
\usepackage{subcaption}
\usepackage{tabularx}
\usepackage{booktabs}

\usepackage{comment}

\usepackage[colorlinks=true,linkcolor=blue,citecolor=blue,urlcolor=blue]{hyperref}

\usepackage[numbers,sort&compress]{natbib}

\newcolumntype{L}[1]{>{\raggedright\arraybackslash}p{#1}}
\newcolumntype{C}[1]{>{\centering\arraybackslash}p{#1}}
\newcolumntype{Y}{>{\raggedright\arraybackslash}X}

\newtheorem{theorem}{Theorem}

\newtheorem{prop}{Proposition}
\newtheorem{lemma}{Lemma}
\newtheorem*{example}{Example}

\theoremstyle{definition}
\newtheorem{definition}{Definition}

\theoremstyle{remark}
\newtheorem{remark}{Remark}

\title{\bf On the Mathematical Foundations of\\
Gravitational Decoupling in General Relativity}

\author{Andrea Giusti\thanks{E-mail: andrea.giusti9@unibo.it}
\\
\\
{\em DIFA, University of Bologna, I-40126 Bologna, Italy}
\\
\\
{\em I.N.F.N., Sezione di Bologna, I.S.~FLAG, I-40127 Bologna, Italy}
\\
\\
{\em Alma Mater Research Center on Applied Mathematics $({\cal AM}^2)$,}\\ 
{\em University of Bologna, I-40123 Bologna, Italy}
}

\date{ }

\begin{document}

\maketitle

\begin{abstract}
\noindent We study the algebraic and geometric structures of gravitational decoupling within a fixed class of globally static, spherically symmetric spacetimes with non-vanishing areal-radius differential. We show that (minimal and extended) geometric deformations of the seed metric define a groupoid of partial transformations. We also draw a connection between the geometric deformation procedures and disformal transformations. It is shown that every admissible minimal geometric deformation admits a disformal representation using the areal radius as an
auxiliary scalar. For extended deformations, we derive a necessary condition and exhibit a global obstruction that persists under diffeomorphisms preserving the areal radius and the static Killing direction up to constant normalization. Finally, we verify that the disformal construction reproduces the standard minimal geometric deformation source equations and separate conservation law, including anisotropic seeds.
\end{abstract}

\newpage

% ==========================================
\section{Introduction}
\label{sec:intro}
% ==========================================

Gravitational decoupling \cite{Ovalle:2020fuo} is a set of techniques used to generate new exact analytical solutions to the Einstein field equations sourced by non-trivial energy-momentum tensors. This approach was originally proposed by J. Ovalle in the context of brane-world cosmology \cite{MGDog} and was later developed within the framework of General Relativity \cite{MGD}, specifically by means of the so-called Minimal Geometric Deformation (MGD) approach.

Ovalle's proposal was as ingenious as it was simple. General Relativity, due to its strong non-linear nature, does not satisfy the superposition principle. In other words, a new solution to the Einstein field equations cannot, in general, be constructed as a linear combination of known solutions. If we consider, for instance, a perfect fluid as a source of the Einstein field equations, assuming staticity and spherical symmetry, then the field equations reduce to the Tolman--Oppenheimer--Volkoff system once an appropriate matter closure and initial or boundary data are specified \cite{Tolman:1939jz,Oppenheimer:1939ne}. However, if we then attempt to add a secondary source to the perfect fluid (with its own energy momentum tensor) the equations become overwhelmingly difficult to solve analytically, and one cannot simply solve separately the system independently for the two sources and linearly combine the solutions. Gravitational decoupling, {\em via} MGD or its extensions, circumvents this problem by providing a systematic procedure to split the full systems of field equations into two simpler (often decoupled, specifically in the case of MGD) sectors.

Let us be more specific and summarize the general machineries as they are currently presented in the literature. To this end we rely on the foundational works by J. Ovalle \cite{MGD} and R. Casadio, {\em et al.} \cite{EGD}. The MGD approach, assuming staticity and spherical symmetry, starts off with a {\em seed geometry}, sourced by a perfect fluid, upon which a secondary, more complex, matter source is added. The main assumption of the method is that this new matter only affects the spatial part of the spacetime geometry (by the addition of a deformation function), specifically stretching or compressing radial distances, while leaving the time component unaltered. The Einstein field equations then split into two distinct sectors, specifically one system of equation for the seed matter and a second system for the newly added source. Notably, since the temporal component of the spacetime metric is left unchanged by the procedure, the standard hydrostatic balance of the seed matter is preserved. Hence, the stress-energy tensors of the seed and newly added matter sources are separately covariantly conserved, meaning that there is no exchange of energy and momentum between the two sectors. However, this split comes at the cost of losing the closure of the resulting system of equations. Indeed, the so-called decoupled $\theta$-sector ({\em i.e.} the one associated with newly added matter source) has one degree of freedom remaining. Hence, one has to impose one additional constraint on the new matter to close the system. This constraint is typically motivated by phenomenological considerations specific to the astrophysical system under investigation.

While the MGD represents a simple and elegant approach, it is often restrictive when dealing with more complicated scenarios. A natural generalization of this approach, known as the  Extended Geometric Deformation (EGD), was proposed by R. Casadio, J. Ovalle, and R. da Rocha in \cite{EGD}. This modified method deforms both the spatial and temporal components of the spacetime metric simultaneously by means of two independent deformation functions. The extension of gravitational decoupling in the presence of matter and the associated exchange between sources were developed further in \cite{Ovalle:2018gic}. The idea underlying EGD later allowed to generalize the technique to axially symmetric systems \cite{Contreras:2021yxe}.

Despite the pragmatic success of the gravitational decoupling paradigm, the MGD and EGD are still predominantly regarded as solution-generating tricks within specific coordinate systems, while the fundamental mathematical structure of this approach remains mostly uncharted territory. The aim of this work is to take a first step towards filling this gap in the literature and to stimulate further research on this topic by taking advantage of the tools of mathematical relativity.

There are earlier mathematical investigations of transformations of spherical geometries. In particular, P.~Boonserm, M.~Visser, and S.~Weinfurtner \cite{Boonserm:2005ni} studied solution-generating maps preserving the perfect-fluid condition, including their composition and inversion. Disformal maps of spherical geometries and the algebraic properties of invertible disformal transformations have also been studied in \cite{Faraoni:2021gdl,Takahashi:2021ttd}.

In this work we isolate the geometric content of gravitational decoupling within a specified standard-static, spherically symmetric class. We show that admissible additive deformations form a groupoid of partial transformations and compare them with scalar disformal transformations.

The manuscript is organized as follows. Section~2 develops the necessary mathematical background for the subsequent algebraic analysis. Section~3 presents the groupoid structure of admissible geometric deformations. Section~4 investigates their relation to disformal transformations and establishes an obstruction to a global disformal representation of extended geometric deformations. Section~5 recovers the standard minimal geometric deformation approach as a particular disformal transformation, choosing the areal radius as the scalar field defining the map.

Throughout the paper, we use units in which $8\pi G=c=1$ and the standard general-relativistic notation of~\cite{Wald:1984rg} and~\cite{Chrusciel:2019lxd}.

% ==========================================
\section{Preliminaries}
\label{sec:pre}
% ==========================================

In this work we restrict our interest to the original scenario in which the MGD and EGD approaches were originally proposed, {\em i.e.} to static, spherically symmetric spacetimes. 

Let $(M,g)$ be a four-dimensional time-orientable connected Lorentzian manifold with metric $g$. This pair is customarily referred to as spacetime. A spacetime is then said to be {\em spherically symmetric} if it admits an effective isometric action of ${\rm SO} (3)$ such that the principal orbits are diffeomorphic to the two-sphere $\mathbb{S}^2$ (see e.g. \cite{Hawking:1973uf,Szenthe:2004a,Szenthe:2004b,Szenthe:2007}). On the principal-orbit region, such a spacetime is locally isometric to a warped product $({Q}, g_Q) \times_r (\mathbb{S}^2, g_{\mathbb{S}^2})$ where $(Q,g_Q)$ is a two-dimensional Lorentzian manifold, $g_{\mathbb{S}^2}$ is the standard unit-round metric on $\mathbb{S}^2$, and $r:Q\to\mathbb{R}^+$ is a smooth function denoting the {\em areal radius}. In other words, the principal-orbit region of $M$ is locally diffeomorphic to $Q \times \mathbb{S}^2$ and the metric reads $g = g_Q + r^2 \, g_{\mathbb{S}^2}$ (see \cite{Kunzle:1967,Wald:1984rg,Clarke:1987, Szenthe:2004a}). At a regular centre, the spherical fibre collapses and the orbit space has a boundary. That said, in this work we consider only the principal-orbit region.

A spacetime is said to be {\em static} if it admits a timelike Killing vector field that is hypersurface-orthogonal (see \cite{Wald:1984rg,Chrusciel:2020fql}). Interestingly, it is possible to show (see {\em e.g.} \cite{Wald:1984rg}) that if a spacetime $(M,g)$ is both {\em static and spherically symmetric}, then it is again locally isometric to a warped product $({Q}, g_Q) \times_r (\mathbb{S}^2, g_{\mathbb{S}^2})$, but the base manifold $Q$ can be further characterized. Specifically, let $K$ be the timelike Killing vector field that is hypersurface-orthogonal, if we choose one of the coordinates on $Q$ to be the parameter $t$ along the integral curves of $K$, and let $\rho$ a spatial coordinate on $Q$, then $K= \partial/\partial t$. Now, the requirement of hypersurface orthogonality implies, by Frobenius' theorem, that the distribution orthogonal to $K$ is integrable. Consequently, one can choose a second coordinate $\rho$ on $Q$ such that the metric contains no mixed $(t,\rho)$ terms. Moreover, the Killing property, together with spherical symmetry, implies that the metric components are independent of $t$ and of the angular coordinates. Thus, adopting the standard spherical chart on $\mathbb{S}^2$ we have that the warped product decomposition of the metric $g$ admits the following form
\begin{equation}
    g =  -\tilde{A}(\rho) \, {\rm d}t \otimes {\rm d}t + \tilde{B}(\rho) \, {\rm d}\rho \otimes {\rm d}\rho + r(\rho)^2 \left( {\rm d}\theta \otimes {\rm d}\theta + \sin^2\theta \; {\rm d}\varphi \otimes {\rm d}\varphi \right)
\end{equation}
with $\tilde{A}(\rho) > 0$ for staticity and $(\theta, \varphi) \in (0, \pi) \times (0, 2\pi)$.

Let $\nabla$ denote the Levi-Civita connection associated with the metric $g$. We can then observe that if $r'(\rho) \neq 0$ on $Q$, then  by the inverse function theorem we can use the areal radius as a valid local coordinate on $Q$, and hence rewrite the coordinate representation of the metric as 
\begin{equation}
\label{eq:sphcoord}
    g =  -{A}(r) \, {\rm d}t \otimes {\rm d}t + {B}(r) \, {\rm d}r \otimes {\rm d}r + r^2 \left( {\rm d}\theta \otimes {\rm d}\theta + \sin^2\theta \; {\rm d}\varphi \otimes {\rm d}\varphi \right)
\end{equation}
with $A(r) > 0$. Note that this case excludes 
scenarios like wormhole throats and direct-product spacetime ({\em e.g.} the Bertotti–Robinson spacetime $({\rm AdS}_2 \times \mathbb{S}^2, g_{{\rm AdS}_2} + L^2 \, g_{\mathbb{S}^2})$ with $L = $ constant, see \cite{Bertotti:1959pf,Robinson:1959ev}). In both cases we either have isolated points where the areal radius reaches a local extremum or entire regions where the areal radius is constant ({\em e.g.} the Bertotti–Robinson spacetime has $r = L = $ constant). 

Although the problem of extending the gravitational decoupling paradigm to wormholes has attracted some attention recently \cite{Tello-Ortiz:2026rvh}, we will limit our discussion to the original setup of both the MGD and EGD, thus excluding these scenarios and assuming the areal radius to be a valid local coordinate on $Q$.

From this moment on we assume ${\rm d} r \neq 0$ on $Q$ and that $Q \cong \mathbb{R} \times I$ with $I \subseteq \mathbb{R}^+$ an open interval ({\em i.e.}, $M$ is globally trivial bundle), which is often the case of interest for physical applications.
\begin{definition}
We denote by $\mathcal{M}(M)$ the set of static, spherically symmetric Lorentzian metric tensors on $M$ admitting a global
canonical representation with $(t,r)\in\mathbb{R}\times I$,
where $I$ is the fixed range of the areal radius. Furthermore, we define $\mathcal{S} := \mathcal{M}(M)/{\rm Diff} (M)$, {\em i.e.} the set of equivalence classes
$$
[g] = \{ \tilde{g} \in \mathcal{M}(M) \,\, | \,\, \exists \phi \in {\rm Diff}(M) \,\, : \,\, \tilde{g} = \phi^* g \} \, ,
$$
where $\phi^* g$ denotes the pullback  of $g$ under the diffeomorphism $\phi$. An element $g_{\rm c} \in [g]$ such that it takes a ``standard form'' 
$$
g_{\rm c} =  -{A}(r) \, {\rm d}t \otimes {\rm d}t + {B}(r) \, {\rm d}r \otimes {\rm d}r + r^2 \, g_{\mathbb{S}^2} \, ,
$$
with respect to the global coordinates $(t, r)$ on $Q = \mathbb{R} \times I$, is called a {\em canonical representative} of the class $[g]$.
\end{definition}

\begin{remark}
    We assume that every class in $\mathcal{S}$ admits a global
    canonical representative. This existence is not guaranteed
    by staticity, spherical symmetry, monotonicity of the areal
    radius, and the global triviality of $M$ as a bundle alone (see {\em e.g.}
    \cite{Sanchez:2007tc}).
    Furthermore, with these specifications the functions $A,B$
    appearing in the standard form of the canonical representative
    are strictly positive.
\end{remark}

Let $C^\infty_+(I) := \{ f \in C^\infty(I, \mathbb{R}) \,\, | \,\, f(x) > 0 \, , \,\,\, \forall x \in I \}$, then we define the equivalence relation on $C^\infty_+(I)$ such that
$$
f \sim g \quad \iff \quad
 \exists\, c \in \mathbb{R}^+ \,\,:\,\, g(x) = c \, f(x) \, , \quad \forall x \in I \, .
$$
We can then introduce the quotient space 
$$
\mathcal{C} \equiv C^\infty_+(I) / \mathbb{R}^+ := \left\{ [f] \,\, | \,\, f \in C^\infty_+(I) \right\} \, ,
$$ 
with $[f] = \{ c \, f \mid c \in \mathbb{R}^+ \}$.

\begin{definition}
We define the map $\Lambda : \mathcal{S} \to \mathcal{C} \times C^\infty_+(I)$ such that
$\forall [g] \in \mathcal{S}$ it is $\Lambda ([g]) = ([A], B)$, with $A,B \in C^\infty_+(I)$ being the functions appearing in the  coordinate representation of a canonical representative of $[g]$.
\end{definition}

\begin{remark}
The quotient $\mathcal{C}$ is required in order to account for time rescalings, {\em i.e.} $t \mapsto c \, t$ with $c > 0$.   
\end{remark}

\begin{prop}
The map $\Lambda : \mathcal{S} \to \mathcal{C} \times C^\infty_+(I)$ is well defined and bijective.
\end{prop}
\begin{proof}
Let us recall that we are assuming that ${\rm d} r \neq 0$ on $Q = \mathbb{R} \times I$ and $M = \mathbb{R} \times I \times \mathbb{S}^2$.

\noindent {\em $\Lambda$ is well defined:} let $[g] \in \mathcal{S}$, we have to show that if $g_1,g_2 \in [g]$ then $([A_1],B_1)=([A_2],B_2)$. Since $g_1,g_2 \in [g]$ then $\exists \phi \in {\rm Diff} (M)$ such that $\phi^\ast g_1 = g_2$. Furthermore, since $g_1,g_2 \in \mathcal{M}(M)$, they admit coordinate representations of the form
$$
g_i = -A_i(r_i) {\rm d} t_i \otimes {\rm d} t_i  + B_i(r_i) {\rm d} r_i \otimes {\rm d} r_i + r_i^2 \, g_{\mathbb{S}^2} \quad \text{for } i \in \{1,2\} \, .$$
By Lemma~\ref{lemmaccio} (in Appendix~\ref{AppA}), there exists an isometry $\Phi : (M,g_2)\to(M,g_1)$ mapping canonical spheres onto canonical spheres.
Preservation of their areas implies $r_1\circ\Phi=r_2$.
By the mixed-component argument used in the proof of the lemma,
and composing $\Phi$ with a spherical reflection if necessary,
we may write $\Phi(t_2,r_2,\omega)=\big(T(t_2,r_2),r_2,R \,\omega\big)$, with $\omega \in \mathbb{S}^2$ and $R \in {\rm SO} (3)$, and $r_1\circ\Phi=r_2$, where $R\in {\rm SO}(3)$ is constant. We have that $R^\ast \, g_{\mathbb{S}^2} = g_{\mathbb{S}^2}$, computing $\Phi^*g_1$, using ${\rm d}T= \partial_{t_2} T\,{\rm d}t_2+\partial_{r_2} T\,{\rm d}r_2$ we have that
\begin{equation*}
\begin{split}
\Phi^*g_1 &=-A_1(r_2)\bigl(\partial_{t_2} T\,{\rm d}t_2+\partial_{r_2} T\,{\rm d}r_2\bigr)^{\otimes 2} +B_1(r_2)\,{\rm d}r_2 \otimes {\rm d}r_2 +r_2^2\, g_{\mathbb{S}^2}\\
&=-A_1(r_2) (\partial_{t_2} T)^2\,{\rm d}t_2 \otimes {\rm d}t_2
  -A_1(r_2) \, \partial_{t_2}T \, \partial_{r_2}T\,({\rm d}t_2 \otimes {\rm d}r_2 + {\rm d}r_2 \otimes {\rm d}t_2)\\
&\quad+\bigl(B_1(r_2)-A_1(r_2)(\partial_{r_2} T)^2\bigr)\,{\rm d}r_2 \otimes {\rm d}r_2+r_2^2 \, g_{\mathbb{S}^2}.
\end{split}
\end{equation*}
which, comparing with $g_2$ and imposing $\Phi^*g_1=g_2$, yields $A_1(r_2)(\partial_{t_2}T)^2=A_2(r_2)$, $\partial_{r_2}T=0$, and $B_1(r_2)-A_1(r_2)(\partial_{r_2} T)^2 = B_2 (r_2)$. Hence $T=T(t_2)$. Furthermore, 
$(\partial_{t_2}T)^2= A_2(r_2)/A_1(r_2)$, and since the two sides depend on independent variables, both are constant. We therefore infer that $T(t_2)=ct_2+t_0$
with $c\neq 0$. Finally, combining all the results gives $B_2=B_1$, while $A_2=c^2A_1$. Thus $([A_1],B_1)=([A_2],B_2)$, proving that $\Lambda$ is well defined.

\noindent {\em $\Lambda$ is injective:} let $[g_1], [g_2] \in \mathcal{S}$ such that $\Lambda([g_1]) = \Lambda([g_2])$. Since $g_1,g_2 \in \mathcal{M}(M)$, they admit coordinate representations of the form
$$
g_i = -A_i(r) {\rm d} t_i \otimes {\rm d} t_i  + B_i(r) {\rm d} r \otimes {\rm d} r + r^2 \, g_{\mathbb{S}^2} \quad \text{for } i \in \{1,2\} \, .$$
By definition $\Lambda([g_1]) = \Lambda([g_2])$ implies  $([A_1], B_1) = ([A_2], B_2)$ from which it follows that $B_1 = B_2 \equiv B$ and $A_2 = c A_1$ for some constant $c > 0$. Now, the map $\phi : M \to M$ such that $\phi (t_2,r, \theta, \varphi) = (\sqrt{c} \, t_2, \, r, \theta, \varphi)$ is a diffeomorphism such that
$$
\phi^* g_1 = -A_1(r) {\rm d} (\sqrt{c} \, t_2) \otimes {\rm d} (\sqrt{c} \, t_2)  + B(r) {\rm d} r \otimes {\rm d} r + r^2 \, g_{\mathbb{S}^2} = g_2 \, .
$$
Hence $g_1 \sim g_2$, which implies $[g_1] = [g_2]$. 

\noindent {\em $\Lambda$ is surjective:} let $([A], B) \in \mathcal{C} \times C^\infty_+(I)$. Selecting a representative $A(r) \in [A]$ we can define a metric tensor $g$ on $M$ as in \eqref{eq:sphcoord}. We now have to prove that $g$ is static, spherically symmetric and ${\rm d}r \neq 0$. Since $A(r) > 0$, then the Killing vector $K=\partial_t$ is indeed timelike. Furthermore, it is easy to check that the one-form $\omega = -A(r)\mathrm{d}t$ satisfies the hypersurface-orthogonality condition $\omega \wedge {\rm d}\omega = 0$ (see {\em e.g.} \cite{Wald:1984rg}), thus proving staticity. By construction the metric \eqref{eq:sphcoord} is spherically symmetric, and hence so is $g$. Lastly, under our assumptions $r$ is a global coordinate and hence ${\rm d} r \neq 0$ everywhere. Therefore, $g \in \mathcal{M}(M)$. Evaluating the map on its equivalence class gives $\Lambda([g]) = ([A], B)$, proving surjectivity.
\end{proof}

Essentially, this bijection allows us to rigorously identify the moduli space (namely,  the abstract set of all physically distinct spacetime configurations) of static, spherically symmetric geometries on $M$ with two strictly positive smooth functions of the areal radius, thus parameterizing all geometric configurations modulo diffeomorphisms, free of coordinate ambiguities. This result is not conceptually new, however it is hardly ever framed rigorously in the literature.

From this moment on we will identify any $[g] \in \mathcal{S}$ with its unique functional representation $([A], B) \in \mathcal{C} \times C^\infty_+(I)$, and we will denote an element of $\mathcal{S}$ simply by this pair of functions.
\section{Algebraic properties of Geometric Deformations}
\label{Sec-Alg}
Let $M \cong \mathbb{R} \times I \times \mathbb{S}^2$ and let the areal radius $r$ be such that ${\rm d} r \neq 0$ on $M$.

We can now provide an operatorial characterisation of both the MGD and EGD procedures \cite{MGD,EGD}.

\begin{definition}\label{def:op}
Let $f,h \in C^\infty (I, \mathbb{R})$. We define the {\bf Casadio--Ovalle--daRocha (COdR) extended geometric deformation operator} as the {\em partial map} (see  Definition \ref{def-part})
$$
\mathcal{D}_{(f,h)} : \mathcal{S} \to \mathcal{S} \quad \mbox{such that} \quad \mathcal{D}_{(f,h)} ([g]) \equiv \mathcal{D}_{(f,h)} ([{\rm e}^{\xi }],\mu^{-1}) := ([{\rm e}^{\xi +f}],(\mu + h)^{-1}) \, ,
$$
where $\xi \in C^\infty (I,\mathbb{R})$ and $\mu \in C^\infty_+ (I)$ are such that $\Lambda ([g]) = ([{\rm e}^{\xi }],\mu^{-1})$. In particular, the restriction $\mathcal{D}_{(h)} :=\mathcal{D}_{(0,h)}$ will be referred to as {\bf Ovalle's minimal geometric deformation operator}.
\end{definition}

Some comments are now in order. First, we point out that Definition~\ref{def:op} provides a formalization of the MGD and EGD approaches in terms of operators on the moduli space of static, spherically symmetric geometries
on $M$. Second, in the standard formulation of MGD and EGD the functions $f,h$ typically appear multiplied by a real constant $\alpha$ \cite{MGD,EGD}. Its presence is not mathematically necessary, but rather it serves as a convenient bookkeeping parameter that makes the decoupling structure and the seed-configuration limit explicit. Hence, we will omit the parameter $\alpha$ in this study.

To make the relation to the MGD and EGD approaches more apparent, let us consider $g \in \mathcal{M}(M)$. Such a metric admits a canonical representative of the form
\begin{equation} \label{g-seed}
g =  - {\rm e}^{\xi(r)}\, {\rm d}t \otimes {\rm d}t + \frac{1}{\mu (r)} \, {\rm d}r \otimes {\rm d}r + r^2 \, g_{\mathbb{S}^2} \, ,
\end{equation}
which serves as a representative of the equivalence class $[g] \in \mathcal{S}$. The COdR operator maps this class to a unique class $[\hat{g}] \in \mathcal{S}$ with a canonical representative being
\begin{equation} \label{g-deformed}
\hat{g} =  - {\rm e}^{\xi(r) + f(r)}\, {\rm d}t \otimes {\rm d}t + \frac{1}{\mu (r) +h (r)} \, {\rm d}r \otimes {\rm d}r + r^2 \, g_{\mathbb{S}^2} \, .
\end{equation}
Note that, as mentioned above, in the standard lore of gravitational decoupling {\em via} geometric deformations one has $f \to \alpha f$ and $h \to \alpha h$ for some real bookkeeping parameter $\alpha$. Thus recovering the standard presentation in \cite{MGD,EGD}.

\begin{remark} \label{remark3}
The requirement that ${\rm Im} (\mathcal{D}_{(f,h)}) \subseteq \mathcal{S}$ implies a constraint on the possible choices of $f,h \in C^\infty (I, \mathbb{R})$. Specifically, it must be that $\hat{g} \in \mathcal{M}(M)$. Hence, $(f,h)$ must be chosen in a suitable subset $\mathcal{A}_g$ of $C^\infty (I, \mathbb{R}) \times C^\infty (I, \mathbb{R})$, which will depend on the pair $({\rm e}^{\xi}, \mu^{-1})$ determining the seed metric $g$. Specifically, let $\xi \in C^\infty (I,\mathbb{R})$, $\mu \in C^\infty_+ (I)$ and $[g] \equiv ([{\rm e}^{\xi }],\mu^{-1})$ ({\rm i.e.}, the representative of $[g]$ takes the form in \eqref{g-seed}), then 
$$
\mathcal{A}_{g} := \left\{ (f, h) \in C^\infty(I, \mathbb{R}) \times C^\infty(I, \mathbb{R}) \;\middle\vert{}\; \mu(r) + h(r) > 0, \; \forall r \in I \right\} \, .
$$
\end{remark}

The last remark prevents the set of admissible transformation to form a group. A traditional group requires a universal set of transformations that are globally closed under composition and inversion, regardless of the object they act on. Here instead we have a seed-dependent admissibility condition. However, we can still attempt at a algebraic characterization of geometric deformations by means of the theory of groupoids  \cite[Sec.~19.3]{Voight:2021}. For the reader's convenience we report from \cite[Sec.~19.3]{Voight:2021} the following definitions: 
\begin{definition} \label{def-part}
Let $X, Y$ be sets.
\begin{itemize}
    \item A {\em partial function} $f:X \to Y$ is a function defined on a subset of the domain $X$.
    \item A {\em groupoid} $(G, \ast, (\cdot)^{-1})$ is a set $G$ with a unary operation $(\cdot)^{-1} : G \to G$ and a partial function $\ast: G \times G \to G$ such that $\ast$ and $(\cdot)^{-1}$ satisfy the associativity, inverse, and identity properties whenever they are defined: (a) $\forall a,b,c \in G$ such that if $a\ast b$ and $b\ast c$ are defined, then the products $(a\ast b)\ast c$ and $a\ast(b\ast c)$ must exist and be equal; (b) $\forall a \in G$, $\exists a^{-1} \in G$ such that $a\ast a^{-1}$ and $a^{-1}\ast a$ are defined (but not necessarily equal); (c) $\forall a,b \in G$ such that $a\ast b$ is defined, we have
    $(a\ast b)\ast b^{-1} = a$ and $a^{-1}\ast (a\ast b) = b$.
\end{itemize}
\end{definition}
The algebraic structure of the geometric deformation procedure depends on three objects: the seed geometry, a pair of deformation functions, and a compatibility condition that guarantees the applicability of the algorithm. Hence, a possible algebraic object that keeps track of all these information is the set
$$
G = \left\{ ([g], f,h) \,\, | \,\, [g] \in \mathcal{S} \, , \,\, (f,h) \in \mathcal{A}_g \right\} \, ,
$$
where the order of appearance is meant to represent that the seed metric class $[g]$, is acted upon by the admissible deformation parametrized by the pair $(f,h)$. In a way, $G$ collects all possible valid metric deformations {\em via} EGD.

We now define the partial function
$\ast : G \times G \to G$, out of the composition law of the COdR operator. Specifically, we say that, given $a = ([g_a], f_a,h_a)$ and $b = ([g_b], f_b,h_b)$ in $G$, the quantity $a \ast b$ is {\em defined} if and only if $[g_b] = \mathcal{D}_{(f_a,h_a)}([g_a])$. In that case we define
$$
a \ast b = ([g_a], f_a + f_b,h_a+h_b) \, .
$$
Note that the admissibility condition is automatically satisfied since $[g_b] = \mathcal{D}_{(f_a,h_a)}([g_a])$ and $b \in G$. In other words, $a \ast b$ algebraically implements $(\mathcal{D}_{(f_b,h_b)} \circ \mathcal{D}_{(f_a,h_a)})([g_a]) = \mathcal{D}_{(f_a + f_b,h_a + h_b)}([g_a])$ keeping track of all the admissibility conditions.

Now, let us observe that if $([g],f,h) \in G$ then also $(\mathcal{D}_{(f,h)}([g]),-f,-h) \in G$, indeed we are essentially taking back the admissible deformation applied to $[g]$ through $\mathcal{D}_{(f,h)}$. Furthermore we have that
$$
([g],f,h) \ast (\mathcal{D}_{(f,h)}([g]),-f,-h) = ([g], f-f, h-h) = ([g], 0, 0) \, , 
$$
and also
$$
(\mathcal{D}_{(f,h)}([g]),-f,-h) \ast ([g],f,h) = (\mathcal{D}_{(f,h)}([g]), f-f, h-h) = (\mathcal{D}_{(f,h)}([g]), 0, 0) \, , 
$$
that indeed are defined, yet not necessarily equal. Therefore, we have a well-defined  unary operation $(\cdot)^{-1} : G \to G$ such that $\forall a = ([g],f,h) \in G$ we have that $a^{-1} = (\mathcal{D}_{(f,h)}([g]),-f,-h)$.

It is now possible to prove the following statement.
\begin{theorem}
    The triple $(G, \ast, (\cdot)^{-1})$ forms a groupoid.
\end{theorem}
\begin{proof}
The existence of the {\em inverse} has already been proven above. 

\noindent {\em Associativity}: let $a_1,a_2,a_3 \in G$ with $a_i = ([g_i],f_i,h_i)$; assuming that $a_1\ast a_2$ and $a_2\ast a_3$ are defined implies that $[g_2] = \mathcal{D}_{(f_1,h_1)}([g_1])$ and
$[g_3] = \mathcal{D}_{(f_2,h_2)}([g_2])$, hence the admissibility conditions of $(a_1\ast a_2) \ast a_3$ and $a_1\ast (a_2 \ast a_3)$ are automatically satisfied. Indeed, we find 
$$
a_1 \ast (a_2 \ast a_3) = ([g_1], f_1 + f_2 + f_3, h_1 + h_2 + h_3) = (a_1\ast a_2) \ast a_3 \, ,
$$
and since, by composition, $[g_3] = \mathcal{D}_{(f_2,h_2)}([g_2]) = (\mathcal{D}_{(f_2,h_2)} \circ \mathcal{D}_{(f_1,h_1)}) ([g_1])$, recalling the identification $[g_i] \equiv ([{\rm e}^{\xi_i}],\mu_i ^{-1})$ we have that 
$$\mu_1(r) + (h_1 + h_2 + h_3)(r) = (\mu_1(r) + h_1(r) + h_2(r)) + h_3(r) = \mu_3(r) + h_3(r) > 0 \, , \quad \forall r \in I \, .$$

\noindent {\em Cancellation property:}
let $a = ([g_a], f_a,h_a) \in G$, $b = ([g_b], f_b,h_b) \in G$, and $a \ast b$ be defined. Then, taking advantage of the associativity property we have that
\begin{equation*}
\begin{split}
(a \ast b) \ast b^{-1} &= ([g_a], f_a+f_b-f_b, h_a+h_b-h_b) = ([g_a], f_a, h_a) = a \, , \\
a^{-1} \ast (a \ast b) &= ([g_b], -f_a+f_a+f_b, -h_a+h_a+h_b) = ([g_b], f_b, h_b) = b \, .
\end{split}
\end{equation*}
\end{proof}

\begin{remark}
Distinct elements $([g],f,h)$ and $([g],f+c,h)$ of $G$, with 
$c\in\mathbb{R}\setminus\{0\}$, yield the same deformed metric class, since
$[{\rm e}^{\xi+f+c}]=[{\rm e}^{\xi+f}]$. This does not conflict with the definition of a groupoid.
\end{remark}

This result provides an algebraic characterisation of the geometric deformation part of the MGD and EGD formalism. 

\section{Geometric Deformations as Disformal Transformations}
\label{Sec-DT}
Let $\psi \in C^\infty (M,\mathbb{R})$ be a scalar field. A {\em disformal transformation} \cite{Bekenstein:1992pj, Zumalacarregui:2013pma, Faraoni:2021gdl, BenAchour:2024hbg} is a local redefinition of the spacetime metric tensor field $g$ such that
$$
g \mapsto \tilde{g} = C(\psi,X) \, g + D (\psi,X) \, {\rm d}\psi \otimes {\rm d} \psi \, ,
$$
or in the traditional tensor-component representation, as
$$
g_{\mu\nu} \mapsto \tilde{g}_{\mu\nu} = C(\psi, X) \, g_{\mu\nu} + D(\psi, X) \nabla_\mu \psi \nabla_\nu \psi \, ,
$$
where $X := -(1/2) \, g(\nabla\psi, \nabla\psi)$ denotes the canonical kinetic term of the scalar field $\psi$.

These transformations have recently proven to be valuable solution-generating tools in modified gravity (see, e.g.,~\cite{Zumalacarregui:2013pma, Faraoni:2021gdl, BenAchour:2020wiw}). Hence, it is natural to investigate whether a connection exists with MGD and EGD.

Under the assumptions of staticity and spherical symmetry, together with $M \cong \mathbb{R} \times I \times \mathbb{S}^2$ and $\nabla r \neq 0$ on $M$, we have that the spacetime $(M,g)$ admits a canonical representative for $g$ that reads as in \eqref{g-seed}. The fundamental question therefore becomes: there exists a disformal transformation that outputs the canonical representative $\hat{g}$ of the MGD/EGD deformed metric?

First of all, since both MGD and EGD preserve the area of the orbits of spherical symmetry, we have that the term $r^2 \, g_{\mathbb{S}^2}$ must be left unaltered by the candidate disformal transformation. This means that it must be $C(\psi, X) = 1$. Furthermore, the transformation must yield a spherically symmetric spacetime, hence $\psi = \psi (t,r)$ wherever $D\neq0$. Then, component-wise the problem reduces to 
$$
\hat{g}_{\mu\nu} = g_{\mu\nu} + D(\psi, X) \, \partial_\mu \psi \, \partial_\nu \psi \, ,
$$
namely
\begin{eqnarray}
   - {\rm e}^{\xi(r) + f(r)} &=& \hat{g}_{tt} = g_{tt} + D \, (\partial_t \psi)^2 = -{\rm e}^{\xi(r)} + D \, \dot{\psi}^2 \, , \label{eq-1} \\
   \frac{1}{\mu (r) +h (r)} &=& \hat{g}_{rr} = g_{rr} + D \, (\partial_r \psi)^2 = \frac{1}{\mu(r)} + D \, (\psi')^2 \, , \label{eq-2}\\
   0 &=& \hat{g}_{tr} = g_{tr} + D \, \partial_t \psi \, \partial_r \psi  =  D \, \dot{\psi} \psi' \, . \label{eq-3}
\end{eqnarray}
Starting from the latter equation we have that it must hold pointwise that either $D = 0$ or $\dot{\psi} = 0$ or $\psi' = 0$ (unless we consider the possibility for a further time reparametrization to remove the additional off-diagonal $(t,r)$-term). 

First, $D = 0$ corresponds to the deformationless case, and hence it is trivial. 

If $\dot{\psi} = 0$, then Eq.~\eqref{eq-1} can only hold if $f\equiv 0$, which corresponds to the MGD scenario, provided that
\begin{equation}
\begin{split}
D(\psi, X) &= \frac{1}{[\psi '(r)]^2} \left(\frac{1}{\mu (r) +h (r)} - \frac{1}{\mu (r)} \right) = - \frac{\mu (r)}{2X} \left(\frac{1}{\mu (r) +h (r)} - \frac{1}{\mu (r)} \right)\\
&= - \frac{1}{2X} \left(\frac{\mu (r)}{\mu (r) +h (r)} - 1 \right) \, ,
\end{split}
\end{equation}
where we used the fact that $-2X = g^{\mu \nu} \partial_\mu \psi \partial_\nu \psi = g^{rr} (\psi')^2 = \mu \, (\psi')^2$ if $\psi = \psi (r)$. Note that if we assume $\psi'(r) \neq 0$ on $I$, then it admits a smooth inverse $r = r(\psi)$, hence making $D$ a function only of $\psi$ and $X$, as required. The (local) invertibility of the disformal map requires \cite{Takahashi:2021ttd} (note that the reference uses a different definition of $X$)
$$
C(C-2XD)(C-XC_X+2X^2D_X)\neq0 \, ,
$$
that in this case ($C=1$) reduces to $1-2XD\neq0$ and $1+2X^2D_X\neq0$, for which we have that
$$
1-2XD = 1+2X^2D_X = \frac{\mu}{\mu +h } \bigg|_{r(\psi)} > 0 \, , 
$$
and hence the (local) invertibility condition is satisfied. 

If $\psi'=0$, then Eq.~\eqref{eq-2} is satisfied only for $h \equiv 0$, thus it can never reproduce non-trivial MGD scenarios, however it can introduce anomalous redshift terms within the EGD approach (${\rm e}^{f(r)}$ cannot simply be absorbed in a redefinition of time since that procedure would lead to the appearance of an off-diagonal term in the metric, indeed ${\rm e}^{f(r)/2} \, {\rm d} t$ is not exact).

Let us also briefly investigate the case in which we allow for a global reparametrization on the base manifold $Q = \mathbb{R} \times I$. Hence, in this case we are studying if EGD can be seen as a composition of a disformal transformation and a diffeomorphism on $Q$.

Let $\tilde{g}$ be the intermediate metric obtained by disforming the seed. This implies the set of equations
\begin{eqnarray}
  \tilde{g}_{tt} &=& g_{tt} + D \, (\partial_t \psi)^2 = -{\rm e}^{\xi(r)} + D \, (\dot{\psi})^2 \, , \label{eq-11} \\
   \tilde{g}_{rr} &=& g_{rr} + D \, (\partial_r \psi)^2 = \frac{1}{\mu(r)} + D \, (\psi')^2 \, , \label{eq-22}\\
   \tilde{g}_{tr} &=& g_{tr} + D \, \partial_t \psi \, \partial_r \psi  =  D \, \dot{\psi} \psi' \, . \label{eq-33}
\end{eqnarray}
As expected, a generic disformal transformation that preserves the area of the orbits of spherical symmetry will, in general, introduce off-diagonal terms in the metric. The standard procedure to eliminate such terms requires to perform a transformation such that $t \to \hat{t} = t + \sigma (r)$, with $\sigma (r)$ a function to be determined. Since $\partial_t \hat{t} = 1$, then $\dot{\psi} =\partial_t \psi (t,r) = \partial_{\hat{t}} \psi (t(\hat{t}),r)$. By means of this transformation we have that
\begin{equation}
    \begin{split}
    &\tilde{g}_{tt} \, {\rm d}t \otimes {\rm d}t + \tilde{g}_{tr} \left( {\rm d}{t} \otimes {\rm d} r + {\rm d}{r} \otimes {\rm d} t \right) + \tilde{g}_{rr} \, {\rm d}r \otimes {\rm d}r\\ 
    &=\tilde g_{tt}\,{\rm d}\hat{t}\otimes {\rm d} \hat{t}+\big(\tilde{g}_{tr}-\tilde{g}_{tt}\, \sigma'\big)({\rm d}\hat{t}\otimes{\rm d}r + {\rm d} r\otimes {\rm d}\hat{t})+\big[\tilde{g}_{rr}-2\tilde{g}_{tr} \, \sigma' +\tilde{g}_{tt}(\sigma  ')^2\big]\,{\rm d}r\otimes {\rm d}r \, .
    \end{split}
\end{equation}
To eliminate the off-diagonal term we must choose $\sigma (r)$ such that
\begin{equation}
    \sigma'(r) = \frac{\tilde{g}_{tr}}{\tilde{g}_{tt}} \, ,
\end{equation}
that implies
\begin{equation}
    \begin{split}
    \tilde{g}_{tt} \, {\rm d}t \otimes {\rm d}t + \tilde{g}_{tr} \left( {\rm d}{t} \otimes {\rm d} r + {\rm d}{r} \otimes {\rm d} t \right) + \tilde{g}_{rr} \, {\rm d}r \otimes {\rm d}r 
    =\tilde g_{tt}\,{\rm d}\hat{t}\otimes {\rm d} \hat{t}+\bigg(\tilde{g}_{rr}-\frac{\tilde{g}_{tr}^2}{\tilde{g}_{tt}} \bigg)\,{\rm d}r\otimes {\rm d}r \, .
    \end{split}
\end{equation}
We now want the transformed metric to equal the deformed one obtained {\em via} the EGD procedure. In other words, we have to check if there exists an admissible disformal map such that
$$
\hat{g} = \tilde g_{tt}\,{\rm d}\hat{t}\otimes {\rm d} \hat{t}+\bigg(\tilde{g}_{rr}-\frac{\tilde{g}_{tr}^2}{\tilde{g}_{tt}} \bigg)\,{\rm d}r\otimes {\rm d}r + r^2 \, g_{\mathbb{S}^2} \, .
$$
Namely, we have to verify if we can enforce the following conditions simultaneously, {\em i.e.}
\begin{eqnarray}
   - {\rm e}^{\xi(r) + f(r)} &=&  -{\rm e}^{\xi(r)} + D \, (\dot{\psi})^{2} \, , \label{eq-13} \\
   \frac{1}{\mu (r) +h (r)} &=& \tilde{g}_{rr}-\frac{\tilde{g}_{tr}^2}{\tilde{g}_{tt}} =  \frac{1}{\mu(r)} + D \, (\psi')^2 -\frac{(D \, \dot{\psi} \psi')^2 }{-{\rm e}^{\xi(r)} + D \, (\dot{\psi})^{2}}  \label{eq-23} \, ,
\end{eqnarray}
namely
\begin{eqnarray}
D \, (\dot{\psi})^{2} &=& e^{\xi}\big(1-e^{f}\big) \, , \label{eq-14} \\
\frac{1}{\mu +h} &=& \frac{1}{\mu} + D (\psi ')^2 \, {\rm e}^{-f} \, , \label{eq-24}
\end{eqnarray}
where we inserted Eq.~\eqref{eq-14} into Eq.~\eqref{eq-23} to obtain Eq.~\eqref{eq-24}. Hence we find
\begin{eqnarray}
D \, (\dot{\psi})^{2} &=& e^{\xi}\big(1-e^{f}\big) \, , \label{eq-15} \\
D \, (\psi')^2 &=& -\frac{e^f \, h}{\mu \,(\mu+h)} \, . \label{eq-25}
\end{eqnarray}
This implies 
\begin{equation}
    \label{product}
    (D \, \dot{\psi} \psi')^2 = \frac{e^{\xi+f} \, (e^f-1) \,h}{\mu \, (\mu+h)} \, .
\end{equation}
Since $\mu>0$ and $\mu+h>0$, the non-negativity of $(D \, \dot{\psi} \, \psi')^2$ requires
\begin{equation}
\label{condition}
f(r) \, h(r) \geq0 \, , \qquad \forall r\in I.
\end{equation}
The inequality is strict wherever $D \, \dot{\psi} \, \psi' \neq 0$. This necessary condition imposes an additional constraint beyond the EGD admissibility conditions outlined in Remark~\ref{remark3}. Consequently, geometrically admissible EGD configurations for which $f \, h <0$ somewhere on $I$ cannot be obtained through a purely disformal transformation followed by a coordinate transformation of the form $\hat{t}=t+\sigma(r)$. 

Since EGD preserves the areal radius and the static Killing field of the seed, we seek a representation as a composition of a disformal transformation and a diffeomorphism compatible with this shared structure. Such diffeomorphisms, that preserve the areal radius and identify the static Killing fields up to a constant normalization, are exhausted by $\hat{t} = \alpha t + \sigma(r)$, $\hat{r}=r$, where $\alpha\neq0$ is a constant. Proceeding as above, under this more general transformation the condition \eqref{condition} becomes $(\alpha^2 \, e^f-1) \, h \geq 0$, $\forall r \in I$.

Interestingly, we can provide a straightforward example to solidify our analysis.
\begin{example}
Let us consider $(\mathbb{M},\eta)$ the Minkowski spacetime as seed of an EGD procedure, {\em i.e.}
$$
\mathbb{M} = \mathbb{R} \times (0, + \infty) \times  \mathbb{S}^2 \, , \qquad 
\eta =  - \, {\rm d}t \otimes {\rm d}t +   {\rm d}r \otimes {\rm d}r + r^2 \, g_{\mathbb{S}^2} \, ,
$$
hence $\xi = 0$ and $\mu = 1$. The two functions
$$
f(r) =-\log (r/\ell) \, , \qquad h(r) = 1 , \qquad \ell>0 \, ,
$$
are such that $(f,h) \in \mathcal{A}_\eta$, and therefore are  admissible deformations for the EGD program. The algorithm then yields 
$$
\hat g= -\frac{\ell}{r}\,{\rm d} t^2+\frac12\,{\rm d}r^2+r^2g_{\mathbb S^2}
$$
which is smooth and Lorentzian $\forall r > 0$. However, considering the most general compatibility condition $(\alpha^2 \, e^f-1) \, h \geq 0$, it yields $\alpha^2 \geq r/\ell$, $\forall r \in (0, +\infty)$. Hence, $\nexists \alpha \in \mathbb{R} \setminus \{ 0\}$ such that $\alpha^2 \geq r/\ell$, $\forall r \in (0, +\infty)$. Therefore, this specific EGD deformation of the seed cannot be obtained globally combining a disformal transformation and a diffeomorphism compatible with the required structure.
\end{example}

We can now summarize the main results of this section in the following statement.
\begin{theorem}
\label{theo-2}
The following hold:
\begin{itemize}
    \item[(i)] Every geometrically admissible MGD can be represented by a disformal transformation.
    \item[(ii)] In the same canonical coordinates, a purely disformal transformation cannot reproduce temporal and radial EGD deformations that are both non-zero at the same point.
    \item[(iii)] Not every geometrically admissible EGD admits a global disformal representation, even when composed with a diffeomorphism preserving the areal radius and identifying the static Killing fields up to a constant normalization.
\end{itemize}
\end{theorem}

\section{MGD as a Disformal Transformation}
\label{Sec-MGD-Dis}
For static, spherically symmetric cases, as the one of interest here, we have a possible privileged choice for the scalar field: $\psi \equiv r : M \to \mathbb{R}^+$.

Consider the disformal transformation
$$
\hat{g} = g + D \, {\rm d} r \otimes {\rm d} r \, , \qquad D(r) = \frac{1}{\mu (r)+h (r)} - \frac{1}{\mu (r)}  \, .
$$
Here $D$ is taken to depend only on $\psi=r$. Theorem~\ref{theo-2} guarantees that this transformation explicitly implements on $g$ the MGD approach corresponding to the pair $(0, h) \in \mathcal{A}_g$, {\em i.e.} $\hat{g}$ in \eqref{g-deformed} with $f=0$.

It is therefore not surprising to find that the effect of said disformal transformation onto the Einstein tensor yields
\begin{eqnarray}
    G^t{}_t[\hat{g}] - G^t{}_t[{g}]
    &=& \frac h{r^2}+\frac{h'}r \, , \\
    G^r{}_r[\hat{g}] - G^r{}_r[{g}]
    &=& h\left(\frac1{r^2}+\frac{\xi'}r\right) \, , \\
    G^\theta{}_\theta[\hat{g}] - G^\theta{}_\theta[{g}] &=& \frac h4
\left(2\xi''+(\xi')^2+\frac{2\xi'}r\right)+\frac{h'}4\left(\xi'+\frac2r\right) = G^\phi{}_\phi[\hat{g}] - G^\phi{}_\phi[{g}] \, ,
\end{eqnarray}
namely, the residual effective stress-energy tensor $G^\mu{}_\nu[\hat{g}] - G^\mu{}_\nu[{g}]$ coincides with the ``additional gravitational source'' stress-energy tensor $\theta ^\mu {}_\nu$ of the MGD approach (cf.~\cite[Eqs.~(21)--(23)]{MGD}).

Furthermore, denoting by $\hat{\nabla}$ the connection compatible with $\hat{g}$ 
and observing that (following a direct computation of the Christoffel symbols for both metrics and applying the definition for each covariant derivative)
$$
\hat{\nabla}_\nu G^\nu{}_\mu[g]
=\nabla_\nu G^\nu{}_\mu[g]=0
$$
we find that $\hat{\nabla}_\nu\theta^\nu{}_\mu = 0$ and, in particular
\begin{equation}
\hat{\nabla}_\nu\theta^\nu{}_r
=(\theta^r{}_r)'
+\frac{\xi'}2(\theta^r{}_r-\theta^t{}_t)
+\frac2r(\theta^r{}_r-\theta^\theta{}_\theta) = 0 \, ,
\end{equation}
{\em i.e.} the conservation equation in \cite[Eq.~(24)]{MGD}. This implies the absence of energy–momentum exchange between the seed source and the additional source, whose interaction is purely gravitational. Although the original MGD formulation assumes a perfect-fluid seed, pressure isotropy is not required for the separate conservation established here, consistently with the formulation allowing anisotropic seeds in \cite{Andrade:2023wux}. 
Note that, for a fixed tensor field in $\Gamma(TM\otimes T^*M)$
that is diagonal in the coordinate system in which the seed is given by the canonical representative and whose
components depend only on $r$, the expression for its
covariant divergence is independent of $\mu$ and hence unchanged under MGD.

\section{Conclusions}
The aim of this work is to stimulate further investigations into the mathematical foundations of gravitational decoupling and to encourage a systematic and rigorous treatment of its geometric and algebraic structures.

Within the specified global static class, gravitational decoupling admits
a precise description in terms of admissible partial transformations of metric classes. These transformations form a groupoid by composition.

The disformal formulation of MGD provides a geometric interpretation of its radial deformation and reproduces the familiar source and conservation equations
without requiring seed pressure isotropy. The situation is more involved for extended deformations. Indeed, an explicit counterexample shows that an admissible EGD might not admit a global scalar disformal representation compatible with the shared areal radius and static Killing direction.

Extensions of the presented analyses across horizons, critical points of the areal radius, and additional matter constraints remain outside the scope of the present work and are deferred to future investigation.

\bigskip
%
%
%
% ==========================================
\section*{Acknowledgements}
% ==========================================

\noindent The author is grateful to R.~Casadio and R.~da Rocha for useful discussions. 

\noindent A.G. is supported by the Italian Ministry of Universities and Research (MUR) through the grant ``BACHQ: Black Holes and The Quantum'' (grant no. J33C24003220006) and by the INFN project FLAG. This work has been carried out in the framework of activities of the National Group of Mathematical Physics (GNFM, INdAM). 

\medskip

\noindent {\bf AI Disclosure.} OpenAI's GPT--6 Astra (September 2026 version) was used during October 2026 for critical feedback, bibliographic assistance, and discussion of selected proof steps, including those in Appendix~\ref{AppA}; the author retains full responsibility for the manuscript.

\appendix

%========================================
\section{Lemmata} \label{AppA}
%=======================================

\begin{lemma} \label{lemmaccio}
Let $I\subseteq(0,\infty)$ be a nonempty open interval. If two canonical metrics
\[
 g_i=-A_i(r)\,{\rm d} t^2+B_i(r)\,{\rm d} r^2+r^2 \, g_{\mathbb{S}^2},
 \qquad A_i,B_i>0,\qquad i=1,2,
\]
on $M=\mathbb{R}\times I\times\mathbb{S}^2$
are such that $(M,g_1)$ and $(M,g_2)$ are globally isometric, then there exists an isometry identifying their
respective canonical spheres $\{t,r\}\times\mathbb{S}^2$.
\end{lemma}

\begin{proof}
Let $\phi:(M,g_2)\to(M,g_1)$ be a global isometry, so that $\phi^*g_1=g_2$.
Denote by $H_i$ the canonical rotation group, regarded as a subgroup of
${\rm Isom}_0(M,g_i)$ (the component of  the isometry group connected with the identity). The orbits of $H_i$ are precisely the canonical spheres.

We use the following result for connected Lie groups (see \cite[Ch.~XV, Theorem~3.1]{Hochschild} and \cite[Theorem A--C]{Hofmann:1994}): every compact
subgroup is contained in a maximal compact subgroup; maximal compact
subgroups are connected and conjugate.

\medskip

\noindent {\em Step 1: The canonical rotation group is maximal compact.}

Fix one of the metrics and denote it by $g$, and denote
its canonical rotation group by $H$. Let $L$ be a maximal compact
subgroup of ${\rm Isom}_0(M,g)$ containing $H$. By the theorem above,
$L$ is connected. Every orbit $N=L\cdot p=\{\ell(p):\ell\in L\}\subseteq M$, with $p \in M$, is a compact, connected,
embedded submanifold without boundary (a compact group action is proper,
and the orbits of a proper action are embedded
\cite[\S3.2.1 and Proposition~3.41]{AlexandrinoBettiol}).

At a point $q \in N$ where $t|_N$ attains its maximum (it must exist since $t|_N:N\to\mathbb \mathbb{R}$ is continuous on a compact $N$) we have ${\rm d} (t|_N)_q = 0$, hence $\forall v \in T_q N$ it holds that ${\rm d} (t|_N)_q (v) = 0$, and hence $T_q N\subseteq \ker({\rm d}t)_q$.

Since $\ker({\rm d}t)$ is spacelike, $N$ is spacelike at $q \in N$ and hence everywhere, by transitivity of the isometric action of $L$ (being $N$ an orbit of $L$). Moreover, it contains the
sphere $H \cdot p$, so $2\leq \dim (N)\leq 3$.

Observe that $\dim (N)= 3$ is impossible. Indeed, consider the projection $\pi:M \to I\times \mathbb{S}^2$ such that $\pi (p) =\pi(t,r,\omega) =(r,\omega)$. Then $\forall v \in T_p M$ can be written as
$v= \alpha \, \partial_t+ \beta \,\partial_r+X$, with $\alpha,\beta \in \mathbb{R}$ and $X\in T_\omega\mathbb S^2$. The differential then maps tangent vectors into tangent vectors, {\em i.e.}
${\rm  d} \pi_p : T_p M \to T_{(r,\omega)} (I\times \mathbb{S}^2)$ such that ${\rm  d} \pi_p (v) = \beta \,\partial_r+X$ (it cancels the contribution of $\partial_t$). Hence, $\ker ({\rm  d} \pi_p) = \mathbb{R} \, \partial_t$, which implies $\ker\big(d(\pi|_N)_p\big)
=T_pN\cap\mathbb R\,\partial_t
=\{0\}$. Hence $d(\pi|_N)$ is injective, which implies that $\dim (T_p N) \leq \dim (T_{(r,\omega)} (I\times \mathbb{S}^2))=3$. If we assume that $\dim (N)= 3$, then $\dim (T_p N) = 3$, and, as $N$ has no boundary, the inverse function theorem implies that $\pi|_N$ is a local diffeomorphism. Thus $\pi(N)$ would be nonempty, open and compact (image of a compact set through a continuous function), hence also closed (being $I\times \mathbb{S}^2$ Hausdorff and $\pi(N)$ compact), in the connected,
non-compact manifold $I\times\mathbb{S}^2$. But in a connected space the only subsets that are both open and closed are $\emptyset$ and the whole space. Hence, since  $\pi(N) \neq \emptyset$, it must be $\pi(N) = I\times\mathbb{S}^2$. Yet $\pi(N)$ is compact whereas, $I\times\mathbb{S}^2$ is non-compact. We reached a contradiction. Therefore it must be $\dim (N)= 2$.

Since $H\cdot p\subseteq N$ and both are embedded submanifolds of $M$, the inclusion $\iota:H\cdot p\hookrightarrow N$ is smooth. For every $x\in H\cdot p$, we have $T_x(H\cdot p)\subseteq T_xN$. Both tangent spaces have dimension two, so they coincide. Thus, ${\rm d}\iota_x$ is an isomorphism. By the inverse
function theorem, $\iota$ is a local diffeomorphism, and hence each $x\in H\cdot p$ has an open neighborhood in $N$
contained in $H\cdot p$. Hence $H\cdot p$ is open in $N$. Since $N$ is an embedded manifold, it is Hausdorff, hence since $H\cdot p$ is compact in $N$, then $H\cdot p$ is closed. Furthermore, since $N$ is connected and $H\cdot p \neq \emptyset$ is both open and closed in $N$, it must be $N = H\cdot p$. Namely, $L \cdot p = H \cdot p$, $\forall p \in M$.

Every $\ell\in L$ therefore preserves each sphere and has the form $\ell(t,r,\omega)=(t,r,R(t,r) \, \omega)$, with $R(t,r)\in {\rm SO}(3)$. The operator $R$ depends smoothly on $(t,r)$ and connectedness of $L$ excludes
reflections. 

For $a \in \{t,r\}$ and $X\in T_\omega\mathbb{S}^2$, the mixed terms
in $\ell^*g=g$ ($\ell \in L \subset {\rm Isom}_0(M,g_i)$) give
$$
g(\partial_a,X)=0 \quad \Rightarrow \quad
0= (\ell^*g)(\partial_a,X) = g\bigl({\rm d}\ell(\partial_a),{\rm d} \ell(X)\bigr)
=r^2 \, g_{\mathbb{S}^2} \bigl((\partial_a R)\omega,RX\bigr) \, .
$$
Indeed, working at $p \in M$, for $a = t$, consider the curve $\gamma (\lambda) = (t+\lambda,r,\omega)$ e $\gamma'(0)=\partial_t|_p$. Then $\ell(\gamma(\lambda))=(t+\lambda,r,R(t+\lambda,r)\omega)$, hence computing the derivative at $\lambda = 0$ yields $
{\rm d}\ell_p(\partial_t|_p)
=\partial_t|_{\ell(p)}+(\partial_tR)(t,r)\omega$. The same applies to the other vectors, and then the result follows from simple algebra.

The vector $(\partial_a R)\omega$ is tangent to the sphere at $R(t,r)\omega$. Indeed, since $R(t,r)$ is orthogonal
and $|\omega|=1$ (Euclidean norm), we have $|R(t,r)\omega|^2=1$. Differentiating with respect to $a\in\{t,r\}$ while keeping $\omega$ fixed gives
that $(\partial_aR)\omega$ is orthogonal to the radial vector $R\omega$, and therefore belongs to $T_{R\omega}\mathbb S^2$. Furthermore, $RX$ ranges
over the entire tangent plane. Hence we have that $g((\partial_a R) \, \omega, v) = 0$, $\forall p=(t,r,\omega)\in M$ and
$\forall v\in T_{R\omega}\mathbb{S}^2$, meaning that $(\partial_aR)\omega=0$, $\forall \omega \in \mathbb{S}^2$, thus $\partial_aR=0$. As a consequence $R$ is constant ({\em i.e.} canonical ${\rm SO}(3)$ rotations), implying that $\ell\in H$, and therefore $L=H$. 

\medskip

\noindent {\em Step 2: The required isometry follows from conjugacy.}

Conjugation by $\phi$ induces a Lie group isomorphism ${\rm Isom}_0(M,g_2)\rightarrow{\rm Isom}_0(M,g_1)$, with $h\mapsto\phi\circ h\circ\phi^{-1}$. Indeed, if $h^*g_2=g_2$, then
$$
(\phi\circ h\circ\phi^{-1})^*g_1
=(\phi^{-1})^*h^*\phi^*g_1
=(\phi^{-1})^*g_2=g_1.
$$
Thus conjugation by $\phi$ is a group isomorphism, with
inverse given by conjugation by $\phi^{-1}$. Since it is
a homeomorphism fixing the identity, it maps the identity components onto one another.

By the first step,
$H_1$ and $\phi H_2\phi^{-1}$ are both maximal compact subgroups of ${\rm Isom}_0(M,g_1)$. By Theorem~3.1 in \cite[Ch.~XV]{Hochschild}, $\exists \chi\in {\rm Isom}_0(M,g_1)$ such that
$$
 \chi \, \phi \, H_2 \, \phi^{-1} \, \chi^{-1} = H_1 \, .
$$
Setting $\Phi=\chi\circ\phi$, we still have 
$$
\Phi^*g_1
=(\chi\circ\phi)^*g_1
=\phi^*(\chi^*g_1)
=\phi^*g_1
=g_2 \, ,
$$
and
\begin{equation}
\begin{split}
\Phi(H_2\cdot p)
&=\{\Phi(h(p)):h\in H_2\}
=\{(\Phi h\Phi^{-1})(\Phi(p)):h\in H_2\} =\{k(\Phi(p)):k\in H_1\}\\
&=H_1\cdot\Phi(p) \, ,
\end{split}    
\end{equation}
since $\Phi \, H_2 \,\Phi^{-1}=H_1$. Thus $\Phi$ maps each canonical sphere of the second metric onto a canonical sphere of the first, as required.
\end{proof}

%\bigskip
\newpage

% ========================================
\bibliographystyle{utphys}
\bibliography{references}
% ========================================
\end{document}